\documentclass[a4paper,12pt]{article}
\usepackage{amsmath, amsthm, amstext}
\usepackage{amssymb, array, amsfonts, marvosym, mathrsfs}

\theoremstyle{plain}
\newtheorem{T}{Theorem}
\newtheorem{Lemma}{Lemma}
\newtheorem{Prop}{Proposition}
\newtheorem{Corollary}{Corollary}

\renewcommand{\qed}{\nobreak \ifvmode \relax \else
      \ifdim\lastskip<1.5em \hskip-\lastskip
      \hskip1.5em plus0em minus0.5em \fi \nobreak
      \vrule height0.75em width0.5em depth0.25em\fi}

\newcommand{\<}{\langle}
\renewcommand{\>}{\rangle}
\newcommand{\eref}[1]{(\ref{#1})}
\def\re{\eref}
\def\1{1\!\!\!\!1}
\def\RR{{\mathbb R}}

\def\UU{{\mathcal U}}
\def\CC{{\mathbb C}}
\def\OO{\Omega}

\def\al{\alpha}

\def\gr{\beta}
\def\er{\varepsilon}

\def\si{\sigma}
\def\ze{\zeta}

\def\FF{{\mathcal F}}
\def\LL{{\mathcal L}}

\def\BB{{\mathcal B}}
\def\KK{{\mathcal K}}
\def\AA{{\mathcal A}}

\def\dd{\partial}
\def\n{\nabla}

\def\dist{{\rm dist}}

\def\clos{{\rm clos}\,}

\def\te{\theta}
\def\wt{\widetilde}

\def\div{{\rm div}\,}

\def\rot{{\rm curl}\,}

\def\supp{{\rm supp}\,}

\def\Ker{{\rm Ker}\,}
\def\Ran{{\rm Ran}\,}
\def\Dom{{\rm Dom}\,}

\def\eps{\varepsilon}
\def\Ga{\Gamma}
\def\ga{\gamma}

\def\phi{\varphi}

\def\be{\beta}

\def\hOmega{{\widehat \Omega}}

\def\test_field{C^{\infty}_0(\Omega; \mathbb{C}^3)}
\def\htest_field{C^{\infty}_0(\hOmega; \mathbb{C}^3)}
\def\ov{\overline}
\newcommand{\cir}[1]{{#1}_0}

\def\ss{\subset}
\def\sm{\setminus}
\def\ti{\tilde}
\def\wti{\widetilde}

\title{C*-algebras and inverse problem of electrodynamics}
\author
{
M.I.~Belishev
\footnote{St.Petersburg Department of V.A.~Steklov Institute of Mathematics of the Russian Academy of Sciences,
27 Fontanka, St. Petersburg 191023, Russia; belishev@pdmi.ras.ru.
Author is supported by the RFBR grant No. 11-01-00407-a.}, \and
M.N. Demchenko
\footnote{
St.Petersburg Department of V.A.~Steklov Institute of Mathematics of the Russian Academy of Sciences.
27 Fontanka, St. Petersburg 191023, Russia.
Chebyshev Laboratory, St. Petersburg State University, 14th Line, 29b, Saint Petersburg, 199178 Russia;
demchenko@pdmi.ras.ru.
Author is supported by the Chebyshev Laboratory (Department of Mathematics and Mechanics, St. Petersburg State University)  under RF Government grant 11.G34.31.0026.}
}

\begin{document}
\maketitle
\begin{abstract}
We consider the dynamical inverse problem for the Maxwell system
on a Riemannian $3$-manifold with boundary in a time-optimal set-up.
Using BC-method we show that the data of the inverse problem (electromagnetic measurements on the boundary)
determine a $C$*-algebra, which
has a spectrum homeomorphic to a part of the manifold. This part depends on the
duration of measurements.
\end{abstract}


\section{Introduction}
We consider initial boundary-value problem for the Maxwell system
on a smooth compact connected Riemannian $3$-manifold $\OO$ with connected boundary $\Ga:=\dd\OO$
for some $T>0$:
\begin{align}
    & e_t = \rot h, \quad h_t = -\,\rot e, \quad (x,t)\in\OO\times (0,T), \notag\\
    & e\mid_{t=0}\, = h\mid_{t=0}\, = 0, \notag\\
    & e_\te\mid_{\Ga\times[0,T]}\, = f.
    \label{Maxwell}
\end{align}
Here $e$ and $h$ are vector fields (electric and magnetic field respectively)
on the manifold depending on time $t\in[0,T]$,
$(\,\cdot)_t$ denotes a derivative by time $t$,
$(\,\cdot)_\te$ is a tangent part of the field on the boundary.
{\em Boundary control} $f$ is a tangent vector field on the boundary,
depending on time.
If $f$ is smooth and satisfies $\supp f \ss \Ga\times (0, T]$,
then the problem~\re{Maxwell} has unique classical solution.

One can set up an inverse problem, related to
the initial boundary-value problem for the Maxwell system,
consisting in recovering of parameters of the medium
by electro\-mag\-ne\-tic measurements on the boundary.
In particular, one can consider the case when $\OO$ is a domain in Euclidean space,
and Maxwell equations contain coefficients $\er$, $\mu$,
describing electric permittivity and magnetic permeability:
$$
	\er e_t = \rot h, \quad \mu h_t = -\,\rot e. 
$$
Then inverse problem consists in recovering of these coefficients or, for example,
of product $\er\mu$ (if the coofficients are scalar),
which means to recover velocity of propagation of electromagnetic waves in medium.

In this work we assume $\er = \mu = 1$, $\OO$ is assumed to be a Riemannian manifold with boundary,
and our goal is to recover the topology of this manifold
(by the data of the inverse problem we construct a topological space homeomorphic to the manifold).
As data we take a {\em response operator}, acting on the boundary control $f$ by the rule
\begin{equation}
	R^T: f \mapsto -\nu \times h_\te |_{\Ga\times[0,T]},
	\label{RT}
\end{equation}
($\nu$ is an inward unit normal vector on the boundary $\Ga$, $\times$ is a vector product),
where $h$ is a magnetic part of solution of the problem~\re{Maxwell} with control $f$.
Using $R^T$ as data we deal with the {\em dynamical} inverse problem.
Moreover, we consider the {\em time optimal} set-up of dynamical inverse problem:
by the response operator $R^{2T}$ we recover the topology of the boundary layer $\OO^T$
of width $T$ (Theorem~\ref{thecorollary}).
Due to the finiteness of the velocity of electromagnetic wave propagation and
simple kinematic reasons we need a response operator with doubled time $R^{2T}$
to recover the layer $\OO^T$.

The inverse problem for the Maxwell system is a subject of a number of papers.
In many of them rather restrictive assumptions concerning unknown coefficients are made.
In book~\cite{book} solutions of a number of inverse problems arising
in mathematical physics are given under assump\-tion that coefficients defined in a half-space
depend only on the depth. 
Paper~\cite{Rom} is devoted to the time-optimal inverse problem of electrodynamics in a half-space
assuming that unknown coefficients are smooth with respect to
the depth 
and analytical in planes that are parallel to the boundary of the half-space. 
Uniqueness and stablility of the solution is proved.

The most general case without restrictive assumptions on coefficients, except smooth\-ness,
was conidered in~\cite{BG00}, \cite{DAA}.
The result of these works is uniqueness of recovering of scalar $\er$, $\mu$
in time-optimal set-up.

The inverse problem on Riemannian manifold with boundary in time-optimal set-up was considered in~\cite{BD}.
Uniqueness of recovering of $\OO^T$ as Riemannian manifold (up to isometry)
by $R^{2T}$ is proved, which is stronger than Theorem~\ref{thecorollary}.
Our goal is to give another scheme of solving the inverse problem, based on
the connection of inverse problems and the theory of $C$*-algebras.

First this connection was used in~\cite{Belishev2003} to solve the electric impedance tomography problem.
In~\cite{GeomRings} the inverse problem for wave equation on Rieman\-nian manifold with boundary
was solved.
There was shown, that the data of the inverse problem determine some operator algebra,
which yields some topological space homeomorphic to the manifold
(more precisely, to the layer $\OO^T$).
Correspondence between algebras and topological spaces is established in the theory of $C$*-algebras~--
this will be discussed in section~\ref{algebras_sec}.
To solve the inverse problem an operator algebra is constructed using operators, which
are obtained from some operator model of the dynamical system.
This model can be constructed by a response operator using the boundary control method (BC-method)~--
this procedure is described in details in section~\ref{DynSection}.
Note that in~\cite{BD} such a model was obtained with the help of the BC-method as well,
while recovering of the manifold by this model was completely different.

We apply the described scheme to the inverse problem of electrodynamics.
The main difficulty that arises here is that obtained algebra is non-commutative.
But it turns out to be commutative modulo compact operators
and the corresponding quotient algebra yields a topological space homeomorphic to $\OO^T$.
This is formulated more precisely in Theorem~\ref{thetheorem_alg} and Corollary~\ref{FactorSpectrum}.

Authors thank B.A.~Plamenevsky for advice concerning the theory of $C$*-algebras.

\section{Some definitions and notations}

Let $\OO$ be a smooth connected compact $3$-manifold with boundary.
$\< \, \cdot, \cdot \>$ denotes inner product of (complex)
tangent vectors corresponding to the Riemannian structure on $\OO$,
$\dist(\cdot, \cdot)$ is a distant function on $\OO$.

The Riemannian structure determines vector product and operators $\n$, $\rot$, $\div$ 
acting on vector fields.
In arbitrary local coordinates corresponding formulas look as follows
($\epsilon^{jmn}$ is the parity of the permutation $(1,2,3) \to (j,m,n)$):
\begin{align*}
	(u\times v)^j & := \left(\det \{g_{mn}\}\right)^{-1/2} \sum_{k,l,m,n} \epsilon^{jmn} g_{mk} u^k g_{nl} v^l, \\
	(\n\phi)^j & := \sum_k g^{jk} \frac{\dd\phi}{\dd x^k}, \\
	(\rot y)^j & := \left(\det \{g_{mn}\}\right)^{-1/2} \sum_{m,n,k} \epsilon^{jmn} \frac{\dd}{\dd x^m} (g_{nk} y^k), \\
	\div y & := \left(\det \{g_{mn}\}\right)^{-1/2} \sum_{k} \frac{\dd}{\dd x^k}
		\left(\left(\det \{g_{mn}\}\right)^{1/2} y^k\right).
\end{align*}

Further $\n$, $\rot$, $\div$ are understood as generalized operators.
For generalized operators on the manifold we have standard formulas of vector analysis:
\begin{align}
	\div(\phi u) &= \<\n \phi, \ov u\> + \phi\,\div u, \label{divfu} \\
	\div(u\times v) &= \<\rot u, \ov v\> - \<u, \rot \ov v\>, \label{divuv} \\
	\rot(\phi u) &= \n \phi \times u + \phi\,\rot u. \label{rotfu}
\end{align}
In~\re{divfu} and \re{rotfu} function $\phi$ is Lipschitz,
the field $u$ is locally integrable, having locally integrable divergence.
In~\re{divuv} we may suppose that $u$ or $v$ is Lipschitz, and the other field is locally
integrable and has locally integrable $\rot$.

The Maxwell system on the Riemannian manifold is usually written for differential forms
using exterior differential instead of $\rot$.
We describe electric and magnetic fields $e$ and $h$ with vector fields on manifold
and thus the Maxwell equations have the same form as in Eucledian space.

$B(H)$ is the set of linear bounded operators in complex Hilbert space~$H$.

$\KK(H)\ss B(H)$ is the set of compact operators in $H$.
$\KK(H_1; H_2)$ denotes the set of compact operators from $H_1$ to $H_2$.

$C(X)$ is the space of continuous complex functions on a compact topological space $X$.
Define a norm in $C(X)$:
\begin{equation}
	\|f\| := \max_{x\in X} |f(x)|.
	\label{Cnorm}
\end{equation}

$\cir{C}(X)$ is a space of continuous complex functions on topological space $X$, such that
for any $\er>0$ the set $\{x\in X : |f(x)| \ge \er\}$ is compact.
The norm in $\cir{C}(X)$ is defined by the formula~\re{Cnorm}.
If $X$ is an open set in $\RR^n$, then $\cir{C}(X)$ consists of continuous functions in $\ov X$
vanishing on $\dd X$.
Clearly if the space $X$ is compact, then $\cir{C}(X) = C(X)$.

$C^k_c(U)$ is the space of $k$ times continuously differentiable functions with compact support in a subset
$U$ of a smooth manifold.

$(\cdot, \cdot)_H$ denotes inner product in Hilbert space $H$.
In $L_2(U)$ or $\vec L_2(U)$ inner product is denoted by $(\cdot, \cdot)_U$.

Integral of operator-valued function $F : \RR \to B(H)$ over interval $I\ss\RR$
is an operator corresponding to the following bilinear form in $H$
$$
	\int_I ds\, (F(s)\, z, w)_H\,,
$$
providing that the integral exists for all $z, w \in H$ and determines a bounded form.

\section{$C$*-algebras}
\label{algebras_sec}
Here we introduce basic notions of the theory of $C$*-algebras.
See~\cite{Dix}, \cite{Murphy} for more details.

Complex Banach space $\AA$ is called a $C$*-{\em algebra}, if
it is equipped with bilinear mapping (called product)
$$
	\AA \times \AA \to \AA, \quad (a, b) \mapsto ab,
$$
such that:
\begin{align}
	&a(bc) = (ab)c \label{prod1} \\
	&\|ab\| \le \|a\|\, \|b\|, \label{prod2}
\end{align}
and involution
$$
	\AA \to \AA, \quad a \mapsto a^*,
$$
such that:
\begin{align}
	& (\te\, a + \mu\, b)^* = \ov\te\, a^* + \ov\mu\, b^*, \quad \te, \mu\in \CC,  \label{inv1} \\
	& (a^*)^* = a, \label{inv2} \\
	& (ab)^* = b^* a^*, \label{inv3} \\
	& \| a^* a \| = \| a \|^2. \label{inv4}
\end{align}

It follows from the properties of involution that
$\|a^*\| = \|a\|$.
Indeed,
$$
	\|a\|^2 = \|a^* a \| \le \|a^*\| \, \|a\|,
$$
hence
$$
	\|a\| \le \|a^*\|.
$$
Passing to $a^*$ we obtain the reverse inequality:
$$
	\|a^*\| \le \| (a^*)^* \| = \| a \|.
$$

The banach space $B(H)$ of bounded linear operators in a Hilbert space $H$ is an example of $C$*-algebra.
As a product operation and involution one takes a composition and a mapping that maps an operator to its adjoint.

Another example is the space $\cir{C}(X)$ with the norm~\re{Cnorm} and involution $f^* := \ov f$.
$\cir{C}(X)$ is a {\em commutative} $C$*-algebra, i.e. $fg = gf$ for all $f, g\in\cir{C}(X)$.

The set $\BB \ss \AA$ is called $C$*-subalgebra of $C$*-algebra $\AA$, if
$\BB$ is a closed subspace of $\AA$ such that $\BB^* = \BB$ and every pair $b, b'\in\BB$ satisfies
$b\,b' \in \BB$.
Clearly that in this case $\BB$ is $C$*-algebra.

Let $Y$ be a subset of $C$*-algebra $\AA$.
There is a minimal $C$*-subalgebra $\BB \ss \AA$ containing $Y$
(it is an intersection of all $C$*-subalgebras containing $Y$).
We say that $\BB$ is a $C$*-{\em subalgebra generated by the set} $Y$.

It can be shown that $C$*-subalgebra $\BB$ generated by the set $Y\ss\AA$
consists of elements of $\AA$ that can be approximated by elements of the following form
\begin{equation*}
	\sum_{i=1}^{N} \te_i\, \prod_{j=1}^{n_i} y_{i,j} \, ,\quad
	\te_i \in \CC, \quad
	y_{i,j} \in Y\cup Y^*.
\end{equation*}

The set $I \ss \AA$ is called a {\em closed two-sided ideal (or simply closed ideal)} of $C$*-algebra $\AA$,
if it is a closed subspace of $\AA$ and
\begin{equation}
	ac, ca \in I \qquad \forall a \in \AA \quad \forall c \in I.
	\label{Ideal}
\end{equation}

Note the following non-trivial fact:
closed ideal is a $C$*-subalgebra in its $C$*-algebra.

An example of a closed ideal is the set $\KK(H) \ss B(H)$ consisting of compact operators.

Let $I$ be a closed ideal in $C$*-algebra $\AA$.
Define product on the quotient space $\AA/I$ as follows.
Let equivalence classes $\ti a, \ti b \in \AA/I$ contain elements $a, b \in \AA$ respectively.
The product $\ti a\, \ti b$ is defined as an equivalence class containing $a b$.
This definition does not depend on choice of $a$ and $b$ since
$$
	a' b' - a\, b = a' \, (b' - b) + (a' - a)\, b \in I
$$
(we used the definition ~\re{Ideal}).
Similarly $\AA/I$ can be equipped with involution.
It can be shown that $\AA/I$ is $C$*-algebra.

Let $\AA, \BB$ be $C$*-algebras.
Linear mapping $\pi: \AA \to \BB$ is called a {\em homomorphism},
if for any $a, a' \in \AA$ we have
\begin{align*}
	& \pi(a a') = \pi(a)\, \pi(a'), \\
	& \pi(a^*) = \pi(a)^*.
\end{align*}
It follows from the properties of product and involution that for any homomorphism we have
$$
	\|\pi(a)\| \le \|a\|.
$$
If homomorphism is one-to-one, then it is called {\em isomorphism}.
The last estimate applied to $\pi$ and $\pi^{-1}$ implies that
every isomorphism of $C$*-algebras preserves the norm.

$C$*-algebras $\AA$ and $\BB$ are said to be {\em isomorphic}, if there exists isomorphism from $\AA$ to $\BB$.
In this case we write
$$
	\AA \simeq \BB.
$$

If $I$ is a closed ideal in $C$*-algebra $\AA$,
then the mapping $\AA \to \AA/I$ that maps an element of $\AA$ to a corresponding equivalence class
is a homomorphism.
This mapping is called a {\em canonical homomorphism} (associated with closed ideal $I$).

Let $\pi: \AA \to \BB$ be a homomorphism of $C$*-algebras.
Then the image $\pi(\AA)$ is $C$*-algebra,
the kernel $\Ker\pi$ is a closed ideal in $\AA$ and we have
\begin{equation}
	\pi(\AA) \simeq \AA / \, \Ker\pi.
	\label{homfactor}
\end{equation}

Next we introduce notion of spectrum of commutative algebra $\AA$.
Let $\AA'$ be a dual space of $\AA$.
Functional $l\in\AA'$ is called multiplicative, if it satisfies
$$
	l(a_1 a_2) = l(a_1)\, l(a_2)
$$
for all $a_1, a_2 \in \AA$
(in other words, $l$ is a homomorphism from $\AA$ to $\CC$).
Denote by $\OO_\AA$ the set of nonzero multiplicative functionals on $\AA$.
This set is non-empty for non-trivial $C$*-algebras.
One can check that $\| l \| = 1$ for $l\in \OO_\AA$,
so $\OO_\AA$ is the subset of the unit sphere in $\AA'$.
With respect to *-weak topology of $\AA'$ the set $\OO_\AA$ is a locally-compact
Hausdorff topological space.
This space is called {\em spectrum} of commutative $C$*-algebra $\AA$.
Note that spectrum is also defined for non-commutative $C$*-algebras.

If $X$ is a locally compact Hausdorff topological space, then the spectrum $\OO_{\cir{C}(X)}$
is homeomorphic to $X$.
In this case multiplicative functionals have the following form
$$
	l_x(f) := f(x), \quad f\in \cir{C}(X),
$$
where $x \in X$.
If $x \ne x'$, then functionals $l_x$ and $l_{x'}$ are different.
Thus we have an injective (and, as can be shown, continuous) imbedding of $X$ to $\OO_{\cir{C}(X)}$.
This can be proved that there are no other multiplicative functionals on $\cir{C}(X)$,
so there is a bijection between $X$ and $\OO_{\cir{C}(X)}$.

The following fact shows the connection between the spectrum of $C$*-algebra and the spectrum
of a self-adjoint operator in a Hilbert space $H$.
Let $L$ be a self-adjoint operator, $\1$ is an identity operator.
$C$*-subalgebra of $B(H)$ generated by $L$ and $\1$ is commutative and its spectrum is homeomorphic
to the spectrum of $L$ endowed with standard topology of $\RR$.

Now we formulate Stone-Weierstrass theorem. 
\begin{T}
Let $X$ be a compact Hausdorff topological space,
$\AA$ is a $C$*-subalgebra of $C(X)$ such that:
for any different points $x, y \in X$ there exists $f\in \AA$ satisfying
$f(x) \ne f(y)$.
Then either $\AA$ coincides with $C(X)$ or $\AA$ consists of all functions 
from $C(X)$ vanishing in some fixed point $x_0\in X$.
\end{T}

\section{Maxwell system}
\label{DynSection}
\subsection{Domains of influence and controllability}
\label{FuncSpaces}

Fix number $T>0$.
$\vec{\LL}_2(\Ga)$ denotes the space of square integrable tangent fields on~$\Ga$.
Introduce the space of controls
$$
    \FF^T := L_2( [0, T]; \vec \LL_2(\Ga) )
$$
with natural inner product
$$
	(f,f')_{\FF^T} := \int_{\Ga\times[0,T]} \<f, f'\>\, d\Ga dt.
$$
In terms of the control theory $\FF^T$ is an {\em exterior space}
with respect to the system~\re{Maxwell}.

Introduce the set $\FF^T_0$ of smooth controls vanishing in the neighborhood of $\Ga\times\{t=0\}$.
The response operator, defined by~\re{RT}, acts in the space $\FF^T$ and is correctly defined on $\FF^T_0$.
Introduce the set of smooth controls supported on the open set $\si\ss\Ga$,
$$
	\FF_{0,\si}^T := \{ f\in\FF_0^T \, |\, \supp f\ss \si\times(0,T]\,\}
$$
and the set of smooth controls delayed for time $T-s$ ($0<s<T$):
$$
  \FF_{0,\si}^{T,s} := \{ f\in\FF_0^T \, |\, \supp f\ss \si\times(T-s,T]\,\}.
$$

For the set $\si\ss\Ga$ and $s>0$ we put
$$
	\OO^s_\si := \{ x\in \OO : \dist(x, \si) < s\}, \quad
	\OO^s := \OO^s_\Ga.
$$
Further we suppose that $T>0$ satisfies the condition
\begin{equation}
	\OO^T \ne \OO.
	\label{T_cond}
\end{equation}

Introduce Hilbert spaces for the set $\si\ss\Ga$ and $s>0$ 
\begin{equation}
    \UU_\si^s := \clos_{\vec L_2}\{\,\rot z\, |\, z\in \vec{C}^\infty(\OO),\,
       \supp z\ss\OO^s_\si \, \}, \quad
    \UU^s := \UU^s_\Ga.
    \label{Udef}
\end{equation}
(in the first definition the support $z$ may contain some part of $\Ga$,
since it may be contained in $\OO^s_\si$).
Every field $y \in \UU^s$ satisfies
\begin{equation}
	\div y = 0.
	\label{divUU}
\end{equation}
Here $\div$ is understood as generalized divirgence in $\OO$, and field $y$ is continued 
by zero outside of $\OO^s$.
Generally speaking, the space $\UU^s$ does not coincide (but is always a subspace of)
with the space of solenoidal fields supported in $\ov{\OO^s}$.
The reason is possible non-trivial topology of the set $\OO^s$.

Define $E^s_\si$ as an orthogonal projection on $\UU^s_\si$ acting in $\UU^T$.
Let $E^s_\si = E^T_\si$ for $s \ge T$ and $E^s_\si = 0$ for $s\le 0$.
In some cases we will consider $E^s_\si$ as projections on $\UU^s_\si$,
acting in the space $\vec L_2(\OO)$ meaning standard imbedding of $\UU^T$ into $\vec L_2(\OO)$.

For electromagnetic wave corresponding to boundary control $f \in \FF^{T, s}_{0, \si}$
for any $t\in(0,T]$ we have 
$$
	e(\cdot,t), h(\cdot,t) \in \UU^{\,t}_\si.
$$
Indeed, the problem~\re{Maxwell} desribes waves propagating with unit speed, hence
the boundary control $f$ acting during time $t$ generates the field supported in $\OO^t_\si$.
Equalities
$$
    e(\cdot, t) = \rot z,\quad
    h(\cdot, t) = \rot z'
$$
for $z, z'$ supported in $\OO^s_\si$,
required in the definition of $\UU^{\,t}_\si$,
can be obtained by integrating the Maxwell equations by time.

Associate with the problem~\re{Maxwell} the {\em control operator} $W^T: f\mapsto e(\cdot, T)$,
acting from the space of controls $\FF^T$ to the space $\UU^T$.
This operator is correctly defined in $\FF^T_0$.
The control operator is unbounded (see~\cite{BG00}) and has closure,
which will still be denoted by $W^T$.

The equality
\begin{equation}
    \clos_{\UU^T} W^T \FF_0^T = \UU^T
    \label{controlable}
\end{equation}
proved in~\cite{BG00} is called {\em approximate controllability} of the system~\re{Maxwell}.
Proof of this fact is based on a vector version of Holmgren-John-Tataru unique continuation theorem
(see \cite{Tat}).
Approximate controllability plays an important role in the dynamical version of the BC-method.

For an open set $\si\ss\Ga$ we have analogue of the equality~\re{controlable}:
\begin{equation}
    \clos_{\UU^T} W^T \FF_{0,\si}^{T,s} = \UU^s_\si.
    \label{controlables}
\end{equation}

\subsection{Operator eikonals}
For an open set $\si\ss\Ga$ we define an {\em operator eikonal}
as element of $B(\UU^T)$:
$$
	I^T[\si] := \int_{[0,T]} (T - s)\, d E^s_\si \, \ge 0.
$$
Let $\AA \ss B(\UU^T)$ be a $C$*-subalgebra, generated by
operator eikonals $I^T[\si]$, where $\si$ ranges over all open subsets of $\Ga$.
Projections $E^s_\si$ corresponding to different sets $\si$ do not commute,
so the same concerns operator eikonals.
Hence $C$*-algebra $\AA$ is non-commutative.

However, we establish the following fact.
\begin{T}
	\label{thetheorem_alg}
	$C$*-algebra $\AA / (\AA \cap\, \KK(\UU^T))$ is commutative and is isomorphic to
	\mbox{$\cir{C}(\OO^T)$.}
\end{T}
Note, that the intersection $\AA \cap\, \KK(\UU^T)$ is clearly a closed ideal in $\AA$, so
$\AA / (\AA \cap\, \KK(\UU^T))$ is $C$*-algebra.
It is unknown whether $\AA$ contains the whole of $\KK(\UU^T)$ or not.

In section~\ref{algebras_sec} we noted that the spectrum of $C$*-algebra of continuous functions
$\cir{C}(X)$ is homeomorphic to $X$.
Therefore, Theorem~\ref{thetheorem_alg} has the following consequence.
\begin{Corollary}
	The spectrum of commutative $C$*-algebra $\AA / (\AA \cap\, \KK(\UU^T))$
	is homeomor\-phic to the space $\OO^T$.
	\label{FactorSpectrum}
\end{Corollary}

The proof of Theorem~\ref{thetheorem_alg} is given in sections~\ref{I_sec}--\ref{algebra}.
The idea of proof is as follows.
Instead of the space $\UU^T$ and operator eikonals $I^T[\si]$ 
consider the space $\vec L_2(\OO^T)$ and operators
\begin{equation}
	\int_{[0,T]} (T - s)\, d X^s_\si,
	\label{simpleint}
\end{equation}
where $X^s_\si$ is an operator that multiplicates by indicator function of $\OO^s_\si$.
It can easily be shown that these integral operators coincide with operators of multiplication by
functions $\max\{ T - \dist(x, \si), 0 \}$.
These operators are self-adjoint and commute with each other, so they generate a commutative $C$*-subalgebra
$\wti\AA$ in $B(\vec L_2(\OO^T))$.
Clearly there is isomorphism between functional $C$*-algebra $\cir{C}(\OO^T)$
and $C$*-subalgebra of $B(\vec L_2(\OO^T))$ consisting
of operators of multiplication by functions from $\cir{C}(\OO^T)$.
This isomorphism maps $C$*-algebra $\wti\AA$ onto $C$*-subalgebra
generated by functions $\max\{ T - \dist(x, \si), 0 \}$.
The latter coincides with $\cir{C}(\OO^T)$~-- this is Lemma~\ref{minTau},
which is proved with the use of Stone-Weierstrass theorem.
Then $\wti\AA$ is isomorphic to $\cir{C}(\OO^T)$.
The proof of Theorem~\ref{thetheorem_alg} is based on the observation that
operators $I^T[\si]$ differ from the operators~\re{simpleint} restricted to the space $\UU^T$ by
compact operators (Lemma~\ref{If}).
From this we conclude that $C$*-algebras $\AA$ and $\wti\AA$ are isomorphic
``up to compact operators'', and then isomorphism $\wti\AA \simeq \cir{C}(\OO^T)$
completes the proof.

Operator eikonals~\re{simpleint} arise in the inverse problem for scalar wave equation.
In~\cite{GeomRings} $\OO^T$ was recovered as a spectrum of commutative $C$*-algebra generated by these operators.
In case of wave equation on graph corresponding operator eikonals generate non-commu\-ta\-tive $C$*-algebra
having very complicated structure, and
there is no analogue of Theorem~\ref{thetheorem_alg}.
Description of the spectrum of this $C$*-algebra is known only for some special graphs.

\subsection{Inverse problem}
\label{subIP}
Here we describe the procedure of recovering of the layer $\OO^T$ as topological space
by $R^{2T}$, based on Theorem~\ref{thetheorem_alg}.

By Theorem~\ref{thetheorem_alg}
the space $\OO^T$ can be obtained as spectrum of commutative $C$*-algebra $\AA / (\AA \cap\, \KK(\UU^T))$.
Although we can not obtain $C$*-algebra $\AA$ and the space $\UU^T$ by the data of the inverse problem,
BC-method allows us to construct a {\em model} Hilbert space $\UU^T_\#$
and $C$*-algebra $\AA_\# \ss B(\UU^T_\#)$ isomorphic to $\AA$ by operator $R^{2T}$.
The spectrum of commutative $C$*-algebra $\AA_\# / (\AA_\# \cap\, \KK(\UU^T_\#))$
is homeomorphic to $\OO^T$.

Construction of model objects $\UU^T_\#$ and $\AA^T_\#$ is based on
{\em connecting form} $c^T$.
This form is defined on the set $\Dom W^T$ as follows:
$$
    c^T[f,f'] := (W^T f, W^T f')_{J^T_\eps},\quad f,f'\in \Dom W^T. 
$$

The following result obtained in \cite{BG00} shows the relation between $c^T$ and the response operator
$R^{2T}$. To formulate it we introduce the operator of odd continuation of controls from the interval
$[0,T]$ to $[0,2T]$:
$$
    (S^T f)(\cdot, t) := \left\{
        \begin{array}{ll} f(\cdot, t), & 0\le t<T,\\
                -f(\cdot, 2T-t), & T\le t\le 2T.
        \end{array} \right.
$$
and define one more set of controls
$$
    \wt\FF^T_0 := \{f\in \FF^T_0\,|\,S^T f\in \FF^{2T}_0\}.
$$
\begin{Prop}
    For any $f\in \wt\FF^T_0$, $f'\in\FF^T_0$ we have
    \begin{equation}
        c^T[f,f'] = \frac{1}{2} \left((S^T)^* R^{2T} S^T f, f'\right)_{\FF^T}.
        \label{ctfg}
    \end{equation}
\end{Prop}
Identities of this type were first obtained by A.S.~Blagoveshchenskii.
Equalities~\re{ctfg} implies the following fact. 
\begin{Prop}
	Operator $|W^T| = ((W^T)^* W^T)^{1/2}$ is uniquely determined by $R^{2T}$.
\end{Prop}

We define the model space, which is a counterpart of $\UU^T$, as follows:
$$
	\UU^T_\# := \FF^T.
$$
Next define operator $\Phi^T: \UU^T_\# \to \UU^T$
as a factor in the polar decomposition
\begin{equation}
    W^T = \Phi^T |W^T|.
    \label{Polar}
\end{equation}
In~\cite{BG00} it was shown that if $T$ satisfies \re{T_cond} the kernel of $W^T$ is trivial.
Then we have
$$
    \clos_{\FF^T} \Ran |W^T| = (\Ker |W^T|)^\perp = (\Ker W^T)^\perp = \FF^T.
$$
Combined with~\re{controlable} this implies that the operator $\Phi^T$ is unitary.

One can say that a pair
\begin{equation}
    \{\, \UU^T_\#,\, |W^T| \, \}
    \label{themodel}
\end{equation}
forms a model of the dynamical system in consideration.
Note, that using a model system instead of the system itself is traditional in BC-method.

Further we need model subspaces ($\si\ss\Ga$)
\begin{equation}
    \UU^s_{\si\#} := (\Phi^T)^* \UU^s_\si \ss \UU^T_\#. 
    \label{Umodel}
\end{equation}
In the next computation we use equalities \re{controlables}, \re{Polar}
($\si$ is an open subset of $\Ga$):
\begin{equation}
    \UU^s_{\si\#} = (\Phi^T)^* \UU^s_\si = (\Phi^T)^* \clos_{\UU^T} W^T \FF^{\,T,s}_{0,\si} =
    \clos_{\FF^T} |W^T| \FF^{\,T,s}_{0,\si}.
    \label{Ussimodel}
\end{equation}
Note one important thing: since the operator $|W^T|$ is determined by the data of the inverse problem,
the last relation shows that these data determine subspaces $\UU^s_{\si\#}$.

Using orthogonal projections $E^s_{\si\#}$ on $\UU^s_{\si\#}$ acting in $\UU^T_\#$
we can obtain operators 
$$
	I^T_\#[\si] := \int_{[0,T]} (T - s)\, d E^s_{\si\#},
$$
which are unitarily equivalent to $I^T[\si]$,
since $E^s_{\si\#} = (\Phi^T)^* E^s_\si \Phi^T$ and
$I^T_\#[\si] = (\Phi^T)^* I^T[\si] \Phi^T$.
Clearly that $C$*-algebra $\AA_\#\ss B(\UU^T_\#)$ generated by such operators for different open $\si$
is isomorphic to $\AA$, and so by Theorem~\ref{thetheorem_alg} we obtain that $C$*-algebra
$\AA_\# / (\AA_\# \cap \KK(\UU^T_\#))$ is commutative and
$$
	\AA_\# / (\AA_\# \cap \KK(\UU^T_\#)) \simeq C_0(\OO^T),
$$
hence its spectrum is homeomorphic to $\OO^T$.

We obtain the following result.
\begin{T}
    \label{thecorollary}
    Let $T>0$ satisfy \re{T_cond}.
    Then the operator $R^{2T}$ uniquely (up to homeomor\-phism) determines the topological space $\OO^T$.
\end{T}

The scheme of recovering of $\OO^T$ may be displayed as follows:
$$
	R^{2T} \Rightarrow \{\, \UU^T_\#,\, |W^T| \, \} \Rightarrow
	\{ E^s_{\si\#} \} \Rightarrow \{ I^T_\#[\si] \} \Rightarrow
	\AA_\# \Rightarrow \OO_{\AA_\# / (\AA_\# \cap \KK(\UU^T_\#))}.
$$

\section{Basic lemma}
\label{I_sec}
Further we use the same notation for a continuous function from $\cir{C}(\OO^T)$
and the operator of multiplication by this function.

For $\si\ss\Ga$ define functions in $\OO^T$ as follows
$$
	\tau[\si](x) := \dist(x, \si), \quad
	\ti\tau^T[\si](x) := \max\, \{ T - \tau[\si](x),\, 0 \}.
$$
It is easy to check that these functions are Lipschitz and belong to $\cir{C}(\OO^T)$.

\begin{Lemma}
	For an arbitrary set $\si\ss\Ga$ we have
	$$
		\ti\tau^T[\si] - I^T[\si] \in \KK(\UU^T; \vec L_2(\OO^T)).
	$$
	\label{If}
\end{Lemma}

The proof of Lemma~\ref{If} uses the fact that is usually used for proving of spectral asymptotics 
for Maxwell operator as well as in mathematical hydrodynamics.
It concerns compactness of imbedding of some space to $\vec L_2(\OO^T)$.
Before we introduce this space we give the following definition.
Let the field $z\in\vec L_2(\OO)$ satisfy $\rot z\in \vec L_2(\OO)$.
Following~\cite{Leis}, we say that the field $z$ satisfies the condition
\begin{equation}
	z_\te|_\Ga = 0,    \label{zte}
\end{equation}
if for any field $v\in\vec L_2(\OO)$, such that
$\rot v\in \vec L_2(\OO)$, we have
$$
	(z, \rot v)_\OO = (\rot z, v)_\OO.
$$
It can be shown,
that due to smoothness of the boundary $\Ga$ it is necessary to check this condition only for
$v\in \vec C^\infty(\OO)$.

Now we define the space
$$
	F := \{ u \in \vec L_2(\OO) : \div u \in L_2(\OO),\, \rot u \in \vec L_2(\OO),\, u_\te|_\Ga = 0\}
$$
with the norm
$$
	\|u\|_F^2 := \|u\|_\OO^2 + \|\div u\|_\OO^2 + \|\rot u\|_\OO^2.
$$
Due to the smoothness of the boundary $\Ga$
the following proposition holds true
(it is claimed for Euclidean domain in section~8.4 of~\cite{Leis},
and can easily be generalized for the case of smooth manifolds).
\begin{Prop}
	\label{FL2}
	The imbedding of the space $F$ to $\vec L_2(\OO)$ is compact.
\end{Prop}
Actually, the stronger fact holds true: the space $F$ coincides with vector Sobolev Space  $\vec H^1(\OO)$,
which compactly imbeds to $\vec L_2(\OO)$.
However, the Proposition~\ref{FL2} will be enough for us.

Now we describe the scheme of proof of Lemma~\ref{If}.
We obtain estimates for $L_2$-norms of $\rot$ and divergence of the difference
$\ti\tau^T[\si] u - I^T[\si] u$
by $L_2$-norm of $u\in\UU^T$ (inequalities \re{rot_est}, \re{div_est}),
and establish the boundary condition~\re{zte} on $\Ga$ for this difference.
This means that the field $\ti\tau^T[\si] u - I^T[\si] u$ belongs to $F$
with the corresponding estimate of norm,
which implies that operator $\ti\tau^T[\si] - I^T[\si]$ restricted to $\UU^T$ is compact
(by the compactness of imbedding $F$ to $\vec L_2$).

We need some more auxiliary facts.
For $h>0$ we have:
\begin{align*}
	& I^T[\si] = \int_{[-h,T+h]} (T - s)\, d E^s_\si =
	T E^T_\si - \int_{[-h,T+h]} s\, d E^s_\si = \\
	& T E^T_\si - (T+h)\, E^T_\si + \int_{[-h,T+h]} E^s_\si \, ds.
\end{align*}
Passing to a limit $h\to 0+$ we obtain that
\begin{equation}
	I^T[\si] = \int_0^T E^s_\si \, ds.
	\label{Imax}
\end{equation}

Let $X^s_\si$ be an operator of multiplication by the indicator function of the set $\OO^s_\si$,
acting in $\vec L_2(\OO^T)$.
The following equality holds true.
\begin{equation}
	X^s_\si (s - \tau[\si]) = \max\,\{s - \tau[\si], 0 \} =  \int_0^s X_\si^\xi \, d\xi.
	\label{tauint}
\end{equation}
Indeed, let $\chi_{\OO^\xi_\si}$ be an indicator function of the set $\OO^\xi_\si$,
for $u,v \in \vec L_2(\OO)$ we have
\begin{align*}
	&\int_0^s (X_\si^\xi u, v)_\OO \, d\xi =
	\int_{\OO\times[0,s]} dx\, d\xi\,\, \chi_{\OO^\xi_\si}(x)\, \<\,  u(x), v(x) \,\> = \\
	&\int_\OO dx\, \<\,  u(x), v(x) \,\> \int_0^s d\xi\, \chi_{\OO^\xi_\si}(x)  =
	\int_\OO dx\, \<\, u(x), v(x) \,\> \, \max\,\{s - \tau[\si](x), 0\}.
\end{align*}
In particular, it follows from \re{tauint} that
\begin{equation}
	\ti\tau^T[\si] = 
	\int_0^T X^s_\si \, ds.	
	\label{titauint}
\end{equation}
Together with~\re{Imax} this implies the relation
\begin{equation}
	(\ti\tau^T[\si] - I^T[\si])\, y = 
	\left(\int_0^T (X^s_\si - E^s_\si) \, ds\right) y, \quad y \in \UU^T.
	\label{diffKT}
\end{equation}
Define a family of operators
$$
	K^s_\si := \int_0^s (X^\xi_\si - E^\xi_\si) \, d\xi,\quad
	s \in [0,T].
$$
It is convinient to consider $K^s_\si$ as operators acting in the space
$\vec L_2(\OO)$, and suppose that fields from the range of $K^s_\si$ are continued by zero to $\OO\sm\OO^s_\si$.

Compactness of the restriction $K^T_\si|_{\UU^T}$ form $\UU^T$ to $\vec L_2(\OO)$
will prove Lemma~\ref{If}.

\begin{Lemma}
    \label{lemmaNz}
    Suppose $\si\ss\Ga$, $s\in (0,T]$
    and the field $\gr \in \vec{L}_2(\OO^s_\si)$
    is smooth in $\OO^s_\si$ (in particular, it is smooth on the boundary
    $\OO^s_\si \cap \Ga$) and is orthogonal to $\UU^s_\si$.
    Then for $z\in\vec{C}^\infty(\OO)$ we have
    $$
        (\gr, K^s_\si\, \rot z)_{\OO^s_\si} = (\gr, \n \tau[\si] \times z)_{\OO^s_\si}.
    $$
\end{Lemma}
\begin{proof}
    Let $0<s'<s$.
    By the absolute continuity of Lebesgue integral we have
    \begin{equation}
    	(\gr, K^{s'}_\si\, \rot z)_{\OO^{s'}_\si} \to
    	(\gr, K^s_\si\, \rot z)_{\OO^s_\si},
    	\quad s' \to s-0.
    	\label{Kslim}
    \end{equation}
    Clearly that $\gr$ is orthogonal to $\UU^\xi_\si$ for $\xi \le s$, therefore
    \begin{align*}
    	&(\gr, K^{s'}_\si\, \rot z)_{\OO^{s'}_\si} = 
    	\int_0^{s'} d\xi\, (\gr, (X^\xi_\si - E^\xi_\si)\, \rot z)_{\OO^\xi} = 
    	\int_0^{s'} d\xi\, (\gr, X^\xi_\si\, \rot z)_{\OO^\xi} =\\
    	&(\gr, (s'-\tau[\si])\, \rot z)_{\OO^{s'}_\si} =
    	((s'-\tau[\si])\, \gr, \rot z)_{\OO^{s'}_\si}
    \end{align*}
    (in the third equality we used~\re{tauint}).
    Choose an open set $U\ss\OO$ with smooth boundary,
    that contains $\ov{\OO^{s'}}$ and is contained with its closure in $\OO^s$.
    Set
    $$
    	h(x) := \max \{s'-\tau[\si](x), 0\}
    $$
    Then
    \begin{equation}
    	((s'-\tau[\si])\, \gr, \rot z)_{\OO^{s'}_\si} = 
    	(h \gr, \rot z)_U.
    	\label{hgrrot}
    \end{equation}
    The field $h\gr$ is Lipschitz, as function $h$ is Lipschitz,
    and the field $\gr$ is smooth in the neighborhood of $\supp h$,
    so we can apply a formula of integration by parts to the obtained inner product.
		Orthogonality of $\gr$ to the space $\UU^s_\si$ implies that
    \begin{equation}
    	\rot\gr \,|_{\OO^{s}_\si} = 0, \quad 
    	\gr_\te|_{\OO^{s}_\si\cap\,\Ga} = 0.
    	\label{rotte}
    \end{equation}
    The second equality implies that in integration by parts integral over $\dd U$ vanishes,
    since $\gr_\te = 0$ on $\dd U \cap \Ga$ and $h = 0$ on $\dd U \sm \Ga$ as $\dd U \sm \Ga\ss U \sm \OO^{s'}$.
    Applying the first equality in~\re{rotte} and formula~\re{rotfu} we obtain:
    \begin{align*}
    	&(h \gr, \rot z)_U =
    	(\rot (h \gr), z)_U =
    	(\n h\times \gr, z)_U =
    	((-\n \tau[\si])\times \gr, z)_{\OO_\si^{s'}} = \\
    	&(\gr, \n \tau[\si] \times z)_{\OO_\si^{s'}}.
    \end{align*}
    This tends to $(\gr, \n \tau[\si] \times z)_{\OO_\si^s}$ as $s' \to s$.
    Combined with~\re{Kslim} this yields the required equality.
\end{proof}

\begin{Lemma}
		Let $\si\ss\Ga$.
		For the field $z\in\vec{C}^\infty(\OO)$ we have
    \begin{equation}
        \label{krotkrot}
        (K^T_\si\, \rot z, K^T_\si\, \rot z)_\OO =
        2\,(K^T_\si\, \rot z, \n \tau[\si] \times z)_\OO.
    \end{equation}
\end{Lemma}
\begin{proof}
    We have
    \begin{align}
        &
        (K^T_\si \rot z, K^T_\si\rot z)_\OO =
        \int_0^T ds\, ((X^s_\si - E^s_\si)\, \rot z,
        K^T_\si \rot z)_\OO =\notag\\
        &
        \int_0^T ds\, \int_0^T d\xi\,
        ((X^s_\si - E^s_\si)\, \rot z, (X^\xi_\si - E^\xi_\si)
        \, \rot z)_\OO=\notag\\
        &
        2 \int_0^T ds\, \int_0^s d\xi\,
        ((X^s_\si - E^s_\si)\, \rot z, (X^\xi_\si - E^\xi_\si)
        \, \rot z)_\OO =\notag\\
        & 2 \int_0^T ds\,
        ((X^s_\si - E^s_\si)\, \rot z, K^s_\si\, \rot z)_{\OO^s_\si}.
        \label{2Re}
    \end{align}
    Clearly that the field $\gr := (X^s_\si - E^s_\si)\, \rot z$
    is orthogonal to $\UU^s_\si$.
    Moreover, it is smooth in $\OO^s_\si$, since it is solenoidal and satisfies~\re{rotte}.
    So we can apply Lemma~\ref{lemmaNz} to the integrand:
    $$
        ((X^s_\si - E^s_\si)\,\rot z, K^s\,\rot z)_{\OO^s_\si} =
        ((X^s_\si - E^s_\si)\,\rot z, \n \tau[\si] \times z)_{\OO^s_\si}.
    $$
    Substituting this to~\re{2Re}, we obtain
    \begin{align*}
        &(K^T_\si\rot z, K^T_\si\rot z) =
        2 \int_0^T ds\, ((X^s_\si - E^s_\si)\,\rot z, \n\tau[\si] \times z)_{\OO^s_\si} =
        \\
        &2\, (K^T_\si\,\rot z, \n\tau[\si] \times z)_\OO.
    \end{align*}
\end{proof}

By formula~\re{krotkrot} applied to $z\in\vec{C}^\infty(\OO)$ we obtain
\begin{align*}
    &\|K^T_\si \rot z\|^2_{\OO} =
    2\,(K^T_\si \rot z, \n\tau[\si] \times z)_{\OO}\le
    C\,\|K^T_\si \rot z\|_{\OO} \cdot \|z\|_{\OO}.
\end{align*}
Therefore,
\begin{equation}
    \label{krotestimate}
    \|K^T_\si \rot z\|_{\OO} \le C\, \|z\|_{\OO}.
\end{equation}

\begin{Lemma}
\label{corollrotk}
For any field $u\in\vec{L}_{2}(\OO)$ we have
\begin{equation}
    \rot (K^T_\si u)\in\vec{L}_{2}(\OO), \quad
    \|\rot (K^T_\si u)\|_\OO \le C\, \|u\|_\OO.
    \label{rot_est}
\end{equation}
Besides,
\begin{equation}
	(K^T_\si u)_\te |_\Ga = 0.
    \label{te_vanish}
\end{equation}
\end{Lemma}
\begin{proof}
	Let $z\in \vec{C}^\infty(\OO)$.
	As $K^T_\si$ is self-adjoint by~\re{krotestimate} we have
	\begin{align*}
		&|(K^T_\si u, \rot z)_{\OO}| = |(u, K^T_\si \rot z)_{\OO}| \le
		\|u\|_{\OO} \cdot \|K^T_\si \rot z\|_{\OO} \le \\
		&C \|u\|_{\OO} \cdot \|z\|_{\OO}.
	\end{align*}
	Since $z$ is arbitrary this estimate implies~\re{rot_est}.
	As $z$ is not necessarily compactly supported, equality \re{te_vanish} holds true.
\end{proof}

\begin{Lemma}
Let $\si\ss\Ga$.
For any field $u\in \UU^T$ we have
\begin{equation}
	\div (K^T_\si u)\in L_{2}(\OO), \quad
	\|\div (K^T_\si u)\|_\OO \le C\, \|u\|_{\OO^T}.
    \label{div_est}
\end{equation}
\end{Lemma}
\begin{proof}
	By the definition of $K^T[\si]$ we have
	$$
		K^T_\si u = \left(\int_0^T X^s_\si \, ds\right) u - \left(\int_0^T E^s_\si \, ds\right) u.
	$$
	The second term belongs to $\UU^T$ and thus is solenoidal 
	in $\OO$ due to~\re{divUU}.
	By~\re{titauint} the first term equals to $\ti\tau^T[\si] \, u$.
	The function $\ti\tau^T[\si]$ being continued by zero outside of $\OO^T$ is Lipschitz in $\OO$.
	The field $u$ continued by zero outside of $\OO^T$ is solenoidal by~\re{divUU}.
	Then by formula~\re{divfu} we have
	$$
		\div(K^T_\si u) = \div(\ti\tau^T[\si] \, u) = \n\ti\tau^T[\si] \times u.
	$$
	This completes the proof.
\end{proof}

\begin{proof}[Proof of Lemma~\ref{If}]
Suppose $u\in \UU^T$.
It follows from the estimates~\re{rot_est}, \re{div_est} and boundary condition~\re{te_vanish}, that
$$
	\| K^T_\si u \|_F \le \widetilde C \, \| u \|_\OO.
$$
Then by compactness of imbedding $F$ to $\vec L_2(\OO)$ (Proposition~\ref{FL2})
we obtain that $K^T_\si \in \KK(\UU^T; \vec L_2(\OO))$.
In view of \re{diffKT} this completes the proof.
\end{proof}

\section{Homomorphism $\hat\pi$}
\label{hatpi_sec}
For a function $f\in \cir{C}(\OO^T)$
we define an operator $E^T[f] \in B(\UU^T)$ as a composition of operator of multiplication by $f$
acting from $\UU^T$ to $\vec L_2(\OO)$,
and orthogonal projection $E^T$ on $\UU^T$ acting in $\vec L_2(\OO)$:
$$
	E^T[f] y := E^T(f y), \quad y\in \UU^T.
$$
\begin{Lemma}
	\label{Eff}
	For any $f\in \cir{C}(\OO^T)$ we have
	$$
		f - E^T[f] \in \KK(\UU^T; \vec L_2(\OO^T)).
	$$
\end{Lemma}
\begin{proof}
	First we prove Lemma for $f \in C^\infty_c(\OO^T)$.
	
	Choose a finite open cover of the support $\supp f$ such tha every set of this cover is
	diffeomorphic to a ball in case $U_j \cap \Ga = \emptyset$ or to a semiball
	$\{x\in \RR^3 : |x| < 1, x^3 \ge 0\}$ otherwise.
	We may assume that $\ov U_j \ss \OO^T$.
	We can choose a unit partition $\ze_j \in C_c^\infty(U_j)$ such that
	$$
		0 \le \ze_j \le 1, \quad
		\sum_j \ze_j \,\Big|_{\supp f} = 1.
	$$
	Clearly that
	$$
		f - E^T[f] = \sum_j (\ze_j f - E^T[\ze_j f]),
	$$
	and functions $\ze_j f$ belong to $C^\infty_c(U_j)$.
	Thus it is necessary to prove the Lemma for function $f$ supported
	in some open set $U$ diffeomorphic to a ball or a semiball and contaning with its closure in $\OO^T$.
	In this case for any $y\in\UU^T$ we have
	\begin{equation}
		(f y - E^T[f]\, y)|_U = \n p_y, \quad
		p_y\in H^1(U), 
		\label{npy}
	\end{equation}
	and if the set $U$ intersects with $\Ga$, then the following equality holds true
	$$
		p_y|_{U\cap\Ga} = {\rm const}.
	$$
	This can easily be obtained with the help of Helmholtz decomposition in $U$.

	Function $p_y$ in~\re{npy} is uniquely determined up to additive constant,
	which is chosen such that
	\begin{equation}
		p_y|_{U\cap\Ga} = 0,
		\label{pyUGa}
	\end{equation}
	in case $U\cap\Ga \ne \emptyset$, and
	$$
		\int_U p_y \, dx = 0
	$$
	otherwise.
	From the Friedrichs' inequality and Poincar\'e inequality it follows that
	in both cases for some $C$ we have
	$$
		\|p_{y}\|_U \le C \|\n p_{y}\|_U = \| f y - E^T [f]\, y \|_U \le C \|f - E^T [f] \| \cdot \|y\|,
	$$
	therefore, the mapping $y \mapsto p_y$ is continuous from $\UU^T$ to $H^1(U)$.
	
	Now suppose that sequence $y_n$ weakly converges to zero in $\UU^T$.
	The the sequence $p_{y_n}$ weakly converges to zero in $H^1(U)$,
	and due to the compactness of embedding $H^1(U)$ to $L_2(U)$
	this implies that
	\begin{equation}
		\| p_{y_n} \|_U \to 0, \quad n \to \infty.
		\label{pnto0}
	\end{equation}
	Next
	$$
		\|f y_n - E^T[f]\, y_n\|^2_{\OO^T} =
		(f y_n, f y_n - E^T[f]\, y_n)_{\OO^T} =
		(f y_n, \n p_{y_n})_{\OO^T}.
	$$
	In the last equality we used \re{npy} and the inclusion $\supp f \ss U$.
	Integrating by parts in this inner product and
	applying formula~\re{divfu} and equality $\div y_n = 0$ we arrive at
	$$
		(f y_n, \n p_{y_n})_{\OO^T} =
		-\int_U \<\n f, y_n\>\, p_{y_n} \, dx \le
		M \|y_n\|_{\OO^T} \cdot \|p_{y_n}\|_U
	$$
	($M$ depends only on $f$).
	Integral over $\dd U$ vanishes, since $f$ vanishes on $\dd U\sm\Ga$ and in case $U \cap\Ga \ne \emptyset$
	we have~\re{pyUGa}.
	Obtained value tends to zero as norms of $y_n$ are bounded and \re{pnto0} takes place.
	Then in view of the result of previous calculation we obtain
	$$
		\|f y_n - E^T[f]\, y_n\|_{\OO^T} \to 0, \quad n \to \infty,
	$$
	which implies that operator $f - E^T[f]$ is compact.
	
	Now turn to the case of $f\in \cir{C}(\OO^T)$.
	The function $f$ can be approximated in $\cir{C}(\OO^T)$ by functions $f_n \in C^\infty_c(\OO^T)$.
	Operators of multiplication by $f_n$ tend to the operator of multiplication by $f$ with respect to opertor norm.
	Then the operator $f - E^T[f]$ is compact being the limit of compact operators.
\end{proof}

Lemmas~\ref{If} and \ref{Eff} have the following
\begin{Corollary}
	\label{Cor}
	For any $\si\ss\Ga$ we have
	$$
		I^T[\si] - E^T[\ti\tau^T[\si]] \in \KK(\UU^T).
	$$
\end{Corollary}

Let $\pi: B(\UU^T) \to B(\UU^T)/\KK(\UU^T)$ be a canonical homomorphism.
Introduce a mapping
$$
	\hat\pi :\, \cir{C}(\OO^T) \to B(\UU^T)/ \KK(\UU^T), \quad
	\hat\pi (f) := \pi (E^T[f]).
$$

\begin{Lemma}\label{isomorphism}
	The mapping $\hat\pi: \cir{C}(\OO^T) \to B(\UU^T)/ \KK(\UU^T)$
	is an injective homomorphism of $C$*-algebras.
\end{Lemma}
\begin{proof}
We proove the following properties:
$$
	\hat\pi (\al f + \be g) = \al \hat\pi (f) + \be \hat\pi (g),
$$
$$
	\hat\pi (\ov f) = (\hat\pi (f))^*,
$$
$$
	\hat\pi (fg) = \hat\pi (f)\, \hat\pi (g),
$$
$$
	\| \hat\pi (f) \| = \|f\|,
$$
where $f, g \in \,\cir{C}(\OO^T)$, $\al, \be \in \CC$.
First three follow from Lemma~\ref{Eff}.
Consider the third, for example.
We show that
\begin{equation}
	E^T[f]\, E^T[g] - E^T[fg] \in \KK(\UU^T).
	\label{diff_comp}
\end{equation}
By Lemma~\ref{Eff} we have
$$
	E^T[f]\, E^T[g] = (f + K_1)\, E^T[g] =
	f E^T[g] + K = f (g + K_2) + K = fg + \wti K,
$$
where $K_1, K_2, K, \wti K \in \KK(\UU^T, \vec L_2(\OO^T))$.
Applying Lemma~\ref{Eff} to function $fg$, we obtain~\re{diff_comp}.

Now turn to the fourth property.
We can restrict ourselves with smooth $f$, since the mapping $\hat\pi$ is bounded,
which follows from obvious inequality
$$
	\| \hat\pi(f) \| \le \|f\|.
$$
Establish the reverse inequality.
We need to show that for any compact operator $K \in \KK(\UU^T)$
we have
\begin{equation}
	\|E^T[f] + K\| \ge \|f\|.
	\label{IfKf}
\end{equation}
Fix a point $x_0\in \OO^T\sm\Ga$ such that $\n f(x_0) \ne 0$
(we suppose that $f$ does not vanish identically).
Choose a sequence of functions $\phi_j\in C^\infty_c(\OO^T\sm\Ga)$
such that $\supp \phi_j$ tend to $x_0$ as $j\to\infty$.
Introduce fields
$$
	y_j := \n f \times \n\phi_j.
$$
Functions $\phi_j$ can be chosen such that every field $y_j$ does not vanish identically.
Due to~\re{divuv} we have $\div y_j = 0$.
Since $\supp y_j$ tend to $x_0$ as $j\to\infty$,
for sufficiently large $j$ fields $y_j$ belong to $\UU^T$.
Further
$$
	f\, y_j = f \n f \times \n\phi_j = \frac{1}{2} \n (f^2) \times \n\phi_j,
$$
so by \re{divuv} $\div(f y_j) = 0$ and for large $j$ fields $f y_j$ also belong to $\UU^T$, hence
\begin{equation}
	E^T[f] y_j = E^T (f y_j) = f y_j.
	\label{EUf}
\end{equation}
Turn to a normed sequence
$$
	\ti y_j = y_j / \|y_j\|.
$$
Obviously the sequence $\ti y_j$ tends to zero, therefore, $K \ti y_j \to 0$ in $\UU^T$.
Combined with~\re{EUf} this yields
$$
	\|(E^T[f] + K)\, \ti y_j\| = \|f \ti y_j + K \ti y_j\| \to |f(x_0)|, \quad j\to \infty
$$
(we took into account that fields $\ti y_j$ are supported in the neighbourhood of $x_0$).
Since $\|\ti y_j\| = 1$ we arrive at the inequality $\|E^T[f] + K\| \ge |f(x_0)|$.
This takes place for all points $x_0$, where $f$ has nonzero gradient, so
\re{IfKf} holds true.
\end{proof}

\section{Proof of Theorem~\ref{thetheorem_alg}}
\label{algebra}

\begin{Lemma}\label{Lemma_separate}
	If $x, y\in \OO^T$ satisfy
	$$
		\max\,\{T - \dist(x, \ga), 0\} = \max\,\{T - \dist(y, \ga), 0\} \quad \forall \ga\in\Ga,
	$$
	then $x = y$.
\end{Lemma}
\begin{proof}
	Suppose that $\ga_x$ and $\ga_y$ are the nearest points of the boundary respectively to $x$ and $y$
	(may be non-unique).
	Then
	$$
		\dist(x, \ga_x),\, \dist(y, \ga_y) < T.
	$$
	Now it follows from the condition of Lemma that
	$$
		\dist(x, \ga_x) = \dist(y, \ga_x), \quad
		\dist(x, \ga_y) = \dist(y, \ga_y).
	$$
	inequality $\dist(y, \ga_x) \ge \dist(y, \ga_y)$ gives
	$$
		\dist(x, \ga_x) \ge \dist(x, \ga_y),
	$$
	which together with $\dist(x, \ga_x) \le \dist(x, \ga_y)$ implies
	$$
		\dist(x, \ga_x) = \dist(x, \ga_y).
	$$
	Now it follows that $\ga_y$ is also a nearest point of the boundary to $x$.
	Both geodesics that connect $\ga_y$ with $x$ and $y$ are orthogonal to the boundary $\Ga$
	and have (due to the last equality) same length. Therefore, these geodesics coincide and $x=y$.
\end{proof}

\begin{Lemma}\label{minTau}
	$C$*-subalgebra of $\cir{C}(\OO^T)$, generated by the set of functions
	$\ti\tau[\si]$, where $\si$ ranges over all open subsets of $\Ga$,
	coincides with $\cir{C}(\OO^T)$.
\end{Lemma}
\begin{proof}
	We use Stone-Weierstrass theorem (see section~\ref{algebras_sec}).
	First we show that for any 
	$x,y\in \OO^T$, $x \ne y$, there exists function $\ti\tau[\si] \in \mathcal T$,
	such that
	\begin{equation}
		\ti\tau[\si](x) \ne \ti\tau[\si](y).
		\label{minmin}
	\end{equation}
	Let $x,y\in \OO^T\sm\Ga$ (the case when one or both points $x, y$ belong to $\Ga$ is trivial).
	By Lemma~\ref{Lemma_separate} there exists a point $\ga \in \Ga$ such that
	$$
		\max\, \{ T - \dist(x, \ga), 0 \} \ne \max\, \{ T - \dist(y, \ga), 0 \}.
	$$
	Since in this relation we can replace $\ga$ by its sufficiently small neighborhood $\si\ss\Ga$
	relation \re{minmin} holds true.
	
	Next we consider a compactification $X$ of the space $\OO^T$: 
	$$
		X := \OO^T \cup \{\infty\}.
	$$
	Here $\infty$ is an ``infinity point''.
	Using Euclidean topological structure on $\OO^T$ one can
	define a topological structure on $X$, such that $X$ becomes a compact Hausdorff topological space
	(see~\cite{Najmark}). 
	Any function from $\cir{C}(\OO^T)$ can be extended by zero at $\infty$, and the extension belongs to~$C(X)$;
	restriction on $\OO^T$ of any function from $C(X)$ that vanishes at $\infty$, belongs to $\cir{C}(\OO^T)$.
	
	Consider $C$*-subalgebra in $C(X)$, generated by functions $\ti\tau[\si]$ naturally extended to $X$.
	It follows from~\re{minmin} that this $C$*-subalgebra satisfies conditions of Stone-Weierstrass theorem.
	Evidently, this $C$*-subalgebra consists of all continuous functions vanishing at $\infty$.
	So we obtain that $C$*-subalgebra in $\cir{C}(\OO^T)$ generated by functions $\ti\tau[\si]$, 
	coincides with $\cir{C}(\OO^T)$.
\end{proof}

Since $\hat\pi$ is a homomorphism (Lemma~\ref{isomorphism}) the set
\begin{equation}
	\widehat\AA := \hat\pi(\cir{C}(\OO^T)) 
	\label{hAC}
\end{equation}
is a $C$*-subalgebra in $B(\UU^T)/\KK(\UU^T)$ (see section~\ref{algebras_sec}).

According to the description of a $C$*-algebra generated by some set, given in the section~\ref{algebras_sec},
every element $a$ of $C$*-subalgebra $\AA$, generated by the set of operator eikonals $I^T[\si]$,
can be approximated by operators of the following form
\begin{equation}
	b = \sum_{i=1}^{N} \te_i\, \prod_{j=1}^{n_i} I^T[\si_{i,j}],\quad
	\te_i \in \CC.
	\label{polynom}
\end{equation}
By Corollary~\ref{Cor} we have
$$
	b = \sum_{i=1}^{N} \te_i\, \prod_{j=1}^{n_i} E^T[\ti\tau^T[\si_{i,j}]]
	+ K,\quad
	K \in \KK(\UU^T).
$$
Since $\pi$ is homomorphism 
(recall that $\pi$ is a canonical homomorphism from $B(\UU^T)$ to $B(\UU^T)/K(\UU^T)$)
obtained relation implies
$$
	\pi(b) =
	\sum_{i=1}^{N} \te_i\, \prod_{j=1}^{n_i}
	\pi \left(E^T[\ti\tau^T[\si_{i,j}]]\right) = 
	\sum_{i=1}^{N} \te_i\, \prod_{j=1}^{n_i}
	\hat\pi \left(\ti\tau^T[\si_{i,j}]\right). 
$$
As $\hat\pi$ is homomorphism as well
we can permute $\hat\pi$ and product and obtain the equality
\begin{equation}
	\pi(b) = \hat\pi (f),
\end{equation}
where
\begin{equation}
	f = \sum_{i=1}^{N} \te_i\, \prod_{j=1}^{n_i}
	\ti\tau^T[\si_{i,j}] \in \cir{C}(\OO^T).
	\label{fk}
\end{equation}
Now let $\{b_k\}$ be a sequence if elements of the form~\re{polynom},
converging to $a\in\AA$,
$f_k \in \cir{C}(\OO^T)$ are corresponding functions of the form~\re{fk}.
The sequence $\pi(b_k) = \hat\pi(f_k)$ is contained in $\widehat\AA$.
Since $\widehat\AA$ is $C$*-algebra, it contains the limit of this sequence.
Thus $\pi(\AA) \ss \widehat\AA$.

From the other hand, functions $f$ of the form~\re{fk}
are dense in $\cir{C}(\OO^T)$, which follows from Lemma~~\ref{minTau}.
Hence the set of ranges $\pi(b)$, where $b$ ranges over all elements of the form~\re{polynom},
is dense in $\widehat\AA$, and the more $\pi(\AA)$ is dense in $\widehat\AA$.
As the range $\pi(\AA)$ is a $C$*-algebra it is closed and the following equality holds true
$$
	\pi(\AA) = \widehat\AA.
$$

Now applying \re{homfactor} to $C$*-algebra $\AA$ and homomorphism $\pi$ restricted to $\AA$
(the kernel of this restriction is $\AA \cap \KK(\UU^T)$)
we obtain
$$
	\AA / (\AA \cap \KK(\UU^T)) \simeq \pi(\AA) = \widehat\AA = \hat\pi(\cir{C}(\OO^T)). 
$$
Lemma~\ref{isomorphism} guarantees that homomorphism $\hat\pi$ is injective,
so it is an isomorphism between $\cir{C}(\OO^T)$ and its range.
Combined with the previous calculation this implies that
$$
	\AA / (\AA \cap \KK(\UU^T)) \simeq \cir{C}(\OO^T).
$$
This completes the proof of Theorem~\ref{thetheorem_alg}.

\end{document}